\documentclass[11pt]{amsart}
\usepackage{geometry} 
\usepackage{graphicx}
\usepackage{amssymb}
\usepackage{amsfonts}
\usepackage{latexsym}
\usepackage{amsmath}
\usepackage[all]{xy}
\usepackage{hyperref}

\usepackage[applemac]{inputenc}
\usepackage{pdfsync}

\def \Z {\mathbb Z}
\def \N {\mathbb N}

\def \cM {\mathcal{M}}
\def \cC {\mathcal{C}}
\def \cG {\mathcal{G}}
\def \cB {\mathcal{B}}
\def \cA {\mathcal{A}}
\def \cU {\mathcal{U}}

\newcommand{\ac}[1]{\left\{#1\right\}}

\newcommand{\pa}[1]{\left(#1\right)}

\newcommand{\pr}[1]{\mathbb{P}\left(#1\right)}

\newdimen\AAdi%
\newbox\AAbo%
\def\AAk#1#2{\setbox\AAbo=\hbox{#2}\AAdi=\wd\AAbo\kern#1\AAdi{}}%
\newtheorem{defi}{Definition}[section]
\newtheorem{prop}[defi]{Proposition}
\newtheorem{lemm}[defi]{Lemma}
\newtheorem{coro}[defi]{Corollary}
\newtheorem{theo}[defi]{Theorem}

\begin{document}

\title[A non-ergodic PCA with a unique invariant measure]{A
  non-ergodic probabilistic cellular automaton with a unique 
  invariant measure}
\author{Philippe {\sc Chassaing}}
\address{Institut \'Elie Cartan,
Univ. Henri Poincar\'e,
BP 239, 54506 Vandoeuvre-les-Nancy Cedex, France}
\email{chassain@iecn.u-nancy.fr}
\author{Jean {\sc Mairesse}}
\address{LIAFA,
    CNRS et Univ. Paris 7, case 7014, 
75205 Paris Cedex 13, France}
\email{mairesse@liafa.jussieu.fr.}

\begin{abstract}
We exhibit a Probabilistic Cellular Automaton (PCA) on $\{0,1\}^\Z$ with a
neighborhood of size 2 which is non-ergodic although it has a unique invariant
measure. This answers by the negative an old open question on whether uniqueness of the
invariant measure implies ergodicity for a PCA. 
\end{abstract}



\date{\today} 

\maketitle

\section{\bf{Introduction}}

Consider a random process on $\Sigma^{\Z^d}$, where $\Sigma$ is a
finite set, with local interactions and a translation
invariant dynamic. There are two natural instanciations, one with
asynchronous updates of the sites of $\Z^d$, and one with synchronous
updates. In the first case, the model is a continuous time Markov
process, known as a (finite range) {\em Interacting Particle System (IPS)}. In the second
case, the model is a discrete time Markov chain known as a {\em
  Probabilistic Cellular Automaton (PCA)}. 

The relevance of 
IPS in statistical mechanics, as well as in many other
contexts, is well established. Let us mention a couple of motivations
for studying PCA.  First, the investigation of
fault-tolerant 
computational models was the motivation for the Russian school~\cite{toombook,gacs}. Second, PCA appear in 
combinatorial problems related to the enumeration of directed
animals~\cite{LeMa}. 
Third, in the context of the classification of (deterministic) cellular
automata (Wolfram's program),
robustness to random errors can be used as a discriminating
criterion~\cite{FMSE}. 

\medskip

For IPS and PCA, the first question is to study the equilibrium
behavior. An equilibrium is characterized by an {\em invariant measure},
that is a probability measure on the state space which is left
invariant by the dynamic. An invariant measure $\mu$ is {\em
  attractive} if, for any initial condition, the state of the system converges (weakly) to $\mu$
as time goes on. 

By a compactness argument, there always exists 
at least one invariant measure. Therefore, there are, a priori, three
possible situations:

\begin{enumerate}
\item several invariant measures;
\item a unique invariant measure which is not attractive;
\item a unique invariant measure which is attractive. 
\end{enumerate}

In the last case, which corresponds to the nicest possible situation,
the model is said to be {\em ergodic}. Roughly, an ergodic system completely
forgets about its initial condition, while a non-ergodic one 
remembers something forever. 

A classical foundational
question is whether the intermediate case exists. In other words, does
uniqueness of the invariant measure imply convergence to it~? 
For monotone systems, the intermediate case
does not exist. But in general, the question is open. 

For IPS, this question is {\em Open Problem 4} in Chapter 1 of the
classical textbook by Liggett~\cite{ligg85}. In \cite{moun95},
Mountford proves that the intermediate case does not exist for
1-dimensional IPS (that is $d=1$). Quoting 
\cite{moun95}, ``it seems more than plausible that the conclusion
(...) is true in higher dimensions''. However,  the question
remains unsettled.
For PCA, the same question is {\em Unsolved problem 3.4.3} in
Toom~\cite{toom95}, or {\em Unsolved
problem 5.7} in Toom~\cite{toom01}.  

\medskip

In the present paper, we answer the question for PCA by exhibiting a
1-dimensional PCA, model $A$, corresponding to the intermediate case
(Theorem \ref{th-main}). There is a
unique invariant measure of the form $(\mu_0 +\mu_1)/2$ and the PCA
maps $\mu_0$ to $\mu_1$ and $\mu_1$ to $\mu_0$. 
Starting from an initial measure $\mu_0$, the probability
measure of the state of the system is $\mu_0$ at even times and
$\mu_1$ at odd times. Therefore there is no convergence. 

\medskip

Observe that the situations for IPS and PCA are
different: in 1-d, the intermediate case exists for PCA, and not for
IPS. This is 
consistent with the situation for Markov processes on a finite state
space: in discrete time, periodic
phenomena may occur which result in the existence of the intermediate
case; in continuous time, the intermediate case does not exist.  

\medskip

To prove the result for model $A$, we introduce two auxiliary PCA.
The first one, model $B$, corresponds to independently moving
particles annihilating when they meet ($p+p\rightarrow \varnothing$). The second one, model $C$, 
corresponds to independently moving
particles merging when they meet ($p+p\rightarrow p$). 
We compute exactly the evolution of the one-dimensional marginals 
for model $C$  (Theorem \ref{th-density}) and models $A$ and $B$
(Prop. \ref{pr-coupling}) starting from a ``full'' configuration. In
particular, it proves
that the speed of convergence to the
invariant measure is of order $1/\sqrt{n}$ for the three models.


\medskip

Continuous-time versions of models $B$ and $C$ have been studied in
the IPS literature under the names of {\em annihilating random walks} and
{\em coalescing random walks}, respectively,
see~\cite{arra,BrGr,grif79}. The PCA and IPS versions of $B$ and $C$ share
the same features: ergodicity with the invariant
measure being the ``all empty'' Dirac
measure, and with similar and subexponential speed of convergence. In the
IPS setting, the
asymptotic speed of convergence was given by Bramson $\&$
Griffeath~\cite{BrGr} for model $C$, and by Arratia~\cite{arra} for model $B$. 
Also, the coupling between the models $B$
and $C$, that we use in 
Section \ref{se-speedAB}, already appears in Griffeath~\cite[Ch. 3, Sec. 5]{grif79} and in Arratia~\cite{arra} in the continuous-time
setting. The novelty is that we get
exact computations for the PCA models, as opposed to asymptotic
results for the IPS ones. 
At last, let us mention that IPS versions of models $B$ and $C$ on a
{\em finite} set of sites have also been studied, see for instance \cite{ayy} for $B$, 
\cite{MoSq} for $C$, and the references therein.


\section{\bf{Probabilistic Cellular Automaton}}

Let $\Sigma$ be a finite set.
Denote by
$\cM(\Sigma)$ the set of probability measures on $\Sigma$. Let us equip
$X=\Sigma^{\Z}$ with the product topology. Denote by $\cM(X)$ the set of probability measures on $X$ for the
Borelian $\sigma$-algebra. Weak convergence of $(\mu_n)_n$ to $\mu$ is
denoted by $\mu_n \stackrel{w}{\longrightarrow} \mu$. 
Let $K$ be a finite subset of $\Z$ and consider $x\in \Sigma^K$. The
\emph{cylinder} defined by $x$ is the set
\[
* x * =  \Bigl\{u\in \Sigma^\Z, \ \forall k
\in K, u_k=x_k \Bigr\}\:.
\]

Given $k\in \Z$ and 
$V=(v_1,\ldots,v_n)\in \Z^n$, we use the notation $k+V$ for
$(k+v_1,\ldots,k+v_n)$, and the notation $V(K)$ for $\{ i \mid \exists
k \in K, \exists v\in V, i =k+v\}$. 

\medskip

Let us introduce probabilistic cellular automata, restricting
ourselves to 1-dimensional models. 

\begin{defi}\label{de-PCA}
The \emph{alphabet} is a finite set $\Sigma$; the set of \emph{sites} is 
$\Z$. The set of \emph{configurations} is $X=\Sigma^{\Z}$. 
Given $V \in \Z^n$, a \emph{transition function} of \emph{neighborhood} $V$
is a function $f: \Sigma^{V} \rightarrow  \cM(\Sigma)$. 
The \emph{probabilistic cellular automaton} (PCA) $F$ of transition
function $f$ is the application 
$\cM(X) \rightarrow  \cM(X), \ \mu \mapsto \mu F$
defined on cylinders by: $\forall K$, $\forall y \in \Sigma^K$, 
$$
\mu F(* y *)=\sum_{x \in \Sigma^{V(K)}}\mu(* x *)\prod_{k\in K} f((x_i)_{i\in
  k+V}
)(y_k)\:.
$$
\end{defi}

Let us look at how $F$ acts on a Dirac measure $\delta_x$. The value
of all the sites are updated. 
The value $x_k$ of the $k$-th site is
changed into the letter $a\in \Sigma$ with probability $f((x_i)_{i\in
  k+V} )(a)$, {\em independently} of the evolution of the other sites. 

\medskip

By specializing Definition \ref{de-PCA}, we recover two famous models:

\begin{itemize}

\item Assume that $V=\{0\}$, then all the sites behave
  independently. The restriction of the PCA to one site is a Markov
  chain evolving on $\Sigma$. Conversely, any Markov chain on a
  finite state space $E$ can be realized as (a restriction of) a PCA on the alphabet $E$   with neighborhood $V=\{0\}$. 

\item Assume that the transition function $f$ is such that: $\forall u
  \in \Sigma^{V}$, $f(u)$ is a Dirac probability measure. Then we may
  view $f$ as a function $\Sigma^{V} \rightarrow  \Sigma$. We obtain a
  (deterministic) \emph{cellular automaton}. 

\end{itemize}

\medskip

A PCA $F$ may be viewed as a Markov chain on the state space $\Sigma^\Z$. 
Thus we borrow the classical terminology of Markov chains. 

\medskip

\begin{defi}
An \emph{invariant (probability) measure} of $F$ is a
probability measure $\mu \in \cM(X)$ such that $\mu F =
\mu$. The PCA $F$ is \emph{ergodic} if it has a unique invariant
measure which is attractive, i.e. if
\begin{equation}\label{eq-ergod}
\mbox{(i)} \ \bigl[ \exists ! \mu \in \cM(X), \ \mu F =\mu \bigr],
\qquad \mbox{(ii)} \ \bigl[ \forall \nu \in \cM(X), \ \nu F^n
  \stackrel{w}{\longrightarrow} \mu \bigr] \:.
\end{equation}
\end{defi}

\medskip

Consider for a moment a Markov chain on a finite state space with transition matrix
$P$. Let $\cG(P)$ be the graph of the matrix $P$. Classically, we have
\begin{align}\label{eq-finite}
\mbox{(i)}  & \iff  \cG(P) \mbox{ has a unique terminal component} \\
\mbox{(i)+(ii)} & \iff  \cG(P) \mbox{ has a unique terminal component
  which is aperiodic}\:. \nonumber
\end{align}
In particular, uniqueness of the invariant measure does not imply
ergodicity. The simplest example of a non-ergodic Markov chain with a
unique invariant measure is the following:  the state space is
$X=\{0,1\}$ and the transition matrix is
\begin{equation}\label{eq-P}
P = \left[ \begin{array}{cc} 0 & 1 \\ 1 & 0 \end{array} \right] \:.
\end{equation}
The unique invariant measure is $\mu=(\delta_0 + \delta_1)/2$ and for
$\nu= \delta_0$, we do not have $\nu P^n \stackrel{w}{\longrightarrow}
\mu$. 

\medskip

For PCA, it was an open question to know  if (i) implies (ii) in
(\ref{eq-ergod}). The purpose of the present paper is to settle the question by proposing a
non-ergodic PCA with a unique invariant measure. 

\medskip

To get a hint of the difficulty, consider for instance a PCA $F$ with neighborhood $V=\{0\}$. Recall that
each site behaves independently and as a finite Markov
chain $P$. As
recalled in (\ref{eq-finite}), $P$ may satisfy either $[\neg (i)]$, $[(i),\neg
  (ii)]$, or $[(i), (ii)]$. We show in Table \ref{tab:compare} how this gets reflected on the PCA
$F$.

\begin{table}[hbt]
\centering
\begin{tabular}{c|c}
Markov chain $P$ & $\quad$ PCA $F$ $\quad$  \\
\hline 
$\neg (i)$ & $\neg (i)$ \\
$(i),\neg (ii)$ &  $\neg (i)$ \\
$(i), (ii)$ & $(i), (ii)$ \\
& 
\end{tabular}
\caption{Finite Markov chain versus
  ``neighborhood 0 PCA''.}\label{tab:compare}
\end{table}

Let us justify the Table. 
If $\mu$ is an invariant measure of $P$, then the product measure
$\mu^{\otimes \Z}$ is an invariant measure of $F$. Therefore, if $P$
has several invariant measures, the same holds for $F$. 
Assume now that $P$ is ergodic with unique invariant measure
$\mu$. One proves easily that $\mu^{\otimes \Z}$ is attractive, so $F$
is ergodic. Let us concentrate now on the intermediate case for
$P$. If $P$ satisfies $[(i),\neg (ii)]$ then $\cG(P)$ has a unique terminal component which is 
periodic, say of period 2.  Let $(\mu_0 + \mu_1)/2$ be
the unique invariant measure of $P$.  Then $F$ has an
infinite number of invariant measures. Indeed, consider any $(u_i)_{i\in \Z} \in
\{0,1\}^\Z$, and let $(v_i)_{i\in \Z}$ be defined by $v_i=1-u_i$ for
all $i$. Then the probability measure $(\otimes_{i\in \Z} \mu_{u_i} +
\otimes_{i\in \Z} \mu_{v_i})/2$ is clearly an invariant measure of
$F$.

\section{\bf{Statement of the main result}}

\subsection{Model $A$}\label{modeleA}  Consider the PCA $F_A$ on the alphabet
$\Sigma=\{0,1\}$, with neighborhood $V=\{-1,0\}$, and transition
function $a$ defined by:  
\[
a(00)(1)=1/2, \quad a(01)(1) = 0, \quad a(10)(1) = 1, \quad
a(11)(1)=1/2 \:.
\]

\begin{figure}[!h]
\begin{center}
\includegraphics[scale=0.6]{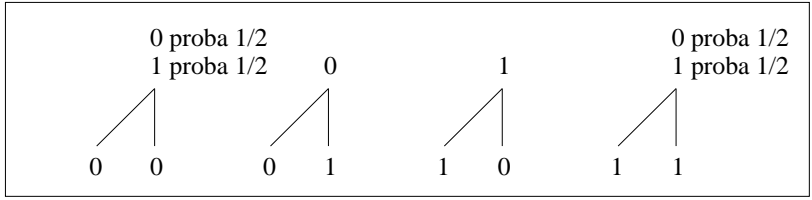}
\caption{The transition function of the PCA $F_A$.}\label{fi-modelA}
\end{center}
\end{figure}

A realization of the Markov chain is obtained as follows.
Consider the function 
\begin{align}\label{eq-A}
\cA: \quad \{0, 1\}^{\Z} \times \cU^{\Z} & \rightarrow  \{0, 1\}^{\Z} \\
(x_i)_{i\in \Z}, (u_i)_{i\in \Z} & \mapsto  (\tilde x_i)_{i\in \Z}\:,\nonumber
\end{align}
with $\cU=\{\uparrow,\rightarrow\}$, and
\[
\tilde x_i = \begin{cases} 0 & \mbox{if } x_{i-1}x_i =01 \mbox{ or } (x_{i-1}x_i, u_i)
  \in \bigl\{ (00, \rightarrow), (11,\uparrow) \bigr\}  \\
1  & \mbox{if } x_{i-1}x_i =10\mbox{ or } (x_{i-1}x_i, u_i)
  \in \bigl\{ (00, \uparrow), (11,\rightarrow) \bigr\} \:.
\end{cases}
\]
Let $U=(U_{i,j})_{(i,j)\in \Z\times \N}$ be a doubly-indexed sequence of i.i.d. r.v.'s with common law
\[\pr{U_{i,j}= \uparrow} = \pr{U_{i,j}=\rightarrow} = 1/2,\]
called the {\em update process}. Set $U_{n}=(U_{i,n})_{i\in \Z}$. Given    a $\{0, 1\}^{\Z}$-valued r.v. $X_{0}=(X_{i,0})_{i\in \Z}$, such that $U\perp X_{0}$, define the sequence of $\{0, 1\}^{\Z}$-valued r.v.'s
$(X_n)_{n\in \N}$ as follows:
\begin{align}\label{rec-A}
X_{n+1} = \cA(X_{n}, U_{n}).
\end{align}
Then $(X_n)_{n\in \N}$ is a realization of model $A$.  The process $U$ is used to randomly update the value of a site, when needed, with $\rightarrow$ being interpreted as ``keep'' and $\uparrow$ as ``switch'', and $X_{i,n}$ is the state of site $i$ at time $n$, so that $X_{n}=(X_{i,n})_{i\in \Z}$ denotes the state of the system at time $n$. 
\subsection{Invariant measure} Let $x=(01)^{\Z}$ be the configuration defined by:  $\forall n\in \Z$, $x_{2n}=0, \ x_{2n+1}=1$. The configuration  $(10)^{\Z}$ is defined similarly. 

\begin{theo}\label{th-main}
The PCA $F_A$ has a unique invariant measure which is
$\mu= (\delta_{(01)^{\Z}} +
\delta_{(10)^{\Z}})/2$. The PCA is non-ergodic. 
\end{theo}

On configurations without $00$ and $11$, the PCA acts as the
translation shift. Therefore $\mu= (\delta_{(01)^{\Z}} +
\delta_{(10)^{\Z}} )/2$ is an invariant measure. 
Assume that it is the unique one. Then $\mu$ is non-attractive, the situation being the same as for (\ref{eq-P}):
consider $\nu= \delta_{(01)^{\Z}}$, then $\nu F_A^n =
\delta_{(01)^{\Z}}$ if $n$ is even, and $\nu F_A^n =
\delta_{(10)^{\Z}}$ if $n$ is odd. 

\medskip

The purpose of Sections \ref{se-aux} and \ref{monotone} is to prove Theorem
\ref{th-main}. 





\section{\bf{Two auxiliary models}}\label{se-aux}

We now define two new PCA, that we call respectively {\em model
  $B$} and {\em model $C$}. 
For both models, the alphabet is $\Sigma= \{\circ, \bullet\}$ and the set of sites
is $\Z$. 
Given a configuration $u \in \{\circ, \bullet\}^{\Z}$, the following interpretation holds: 
if $u_i= \circ$, the site $i$ is ``empty''; if $u_i= \bullet$, the
site $i$ contains a ``particle''. At a given time step, a particle decides (independently of
the others and independently of the past) to remain at its site with
probability 1/2, or to jump to the site on the right with
probability 1/2. In model $B$, if two particles collide, then they
annihilate. In model $C$, if two particles collide, they
are merged into
one particle. Let us define the models more formally.

\subsection{Model $B$.}
\label{modeleB} 
It is the Markov chain on $\{\circ, \bullet\}^{\Z}$ defined as follows. Consider the function 
\begin{align}\label{eq-B}
\cB: \quad \{\circ, \bullet\}^{\Z} \times \cU^{\Z} & \rightarrow  \{\circ, \bullet\}^{\Z} \\
(y_i)_{i\in \Z}, (u_i)_{i\in \Z} & \mapsto  (\tilde y_i)_{i\in \Z}\:, \nonumber
\end{align}
with $\cU=\{\uparrow,\rightarrow\}$, and
\[
\tilde y_i = \begin{cases} \bullet & \mbox{if } (y_{i-1}y_i, u_{i-1}u_i)
  \in \{ (\bullet\circ, \rightarrow \cU), (\circ\bullet, \cU\uparrow),
  (\bullet\bullet, \uparrow\uparrow), (\bullet\bullet,
  \rightarrow\rightarrow) \} \\
\circ & \mbox{otherwise }\:.
\end{cases}
\]
Let $U$ be an update process, defined as in Section \ref{modeleA}. 
Given    a $\{\circ, \bullet\}^{\Z}$-valued r.v. $Y_{0}$, such that $U\perp Y_{0}$, define the sequence of $\{\circ, \bullet\}^{\Z}$-valued r.v.'s $(Y_n)_{n\in \N}$ as follows:
\begin{align}\label{rec-B}
Y_{n+1} = \cB(Y_{n}, U_{n}).
\end{align}
Then $(Y_n)_{n\in \N}$ is a realization of model $B$. 

\medskip

\begin{figure}[!h]
\begin{center}
\includegraphics[scale=0.65]{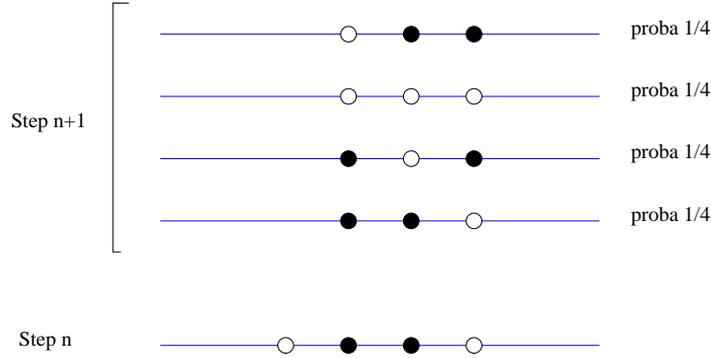}
\caption{The transition function of model $B$.}\label{fi-modelB}
\end{center}
\end{figure}

\medskip

\noindent
{\bf Remarks.}
In the above presentation, model $B$ is a Markov chain with synchronous updates and local
interactions, but not {\em stricto sensu} a PCA. Indeed, if $Y_0$ is
deterministic, then the r.v.'s $Y_{i,1}$ and $Y_{i+1,1}$ are not
  independent, since they are updated using the non-disjoint r.v.'s
  $\{U_{i-1,0},U_{i,0}\}$ and $\{U_{i,0},U_{i+1,0}\}$.
However, it is possible to give a PCA presentation of  model $B$ on a larger alphabet. Define the sequence of $\bigl(\{\circ, \bullet\}\times \mathcal U\bigr)^{\Z}$-valued r.v's $(\widetilde{Y}_n)_{n\in \N}$ by $\widetilde{Y}_n = (Y_n, U_{n})$. We have: 
\[
(\widetilde{Y}_{n+1})_i =
\bigl( \cB(\widetilde{Y}_{n})_i, \ U_{i,n+1} \bigr) \:.
\]
Thus $(\widetilde{Y}_n)_n$ is a realization of a PCA on the alphabet $\{\circ, \bullet\}\times \mathcal U$, with neighborhood $V=\ac{-1,0}$. The same remark holds for model $C$ below. 

The continuous time version of model $B$, with exponential holding times, is called an \textit{annihilating random walk} (cf. \cite[Ch. 3, Sec. 5]{grif79}).

\medskip

\subsection{Model $C$.} \label{modeleC}
It is the Markov chain on $\{\circ, \bullet\}^{\Z}$ defined as
follows. Consider the function 
\begin{align*}
\cC: \quad \{\circ, \bullet\}^{\Z} \times \cU^{\Z} & \rightarrow  \{\circ, \bullet\}^{\Z} \\
(z_i)_{i\in \Z}, (u_i)_{i\in \Z} & \mapsto  (\tilde z_i)_{i\in \Z}\:,
\end{align*}
with
\[
\tilde z_i = \begin{cases} \bullet & \mbox{if } (z_{i-1}z_i, u_{i-1}u_i)
  \in \{ (\bullet\circ, \rightarrow \cU), (\circ\bullet, \cU\uparrow),
  (\bullet\bullet, \uparrow\uparrow), (\bullet\bullet,
  \rightarrow \cU) \} \\
\circ & \mbox{otherwise }\:.
\end{cases}
\]
Let $U$ be an update process. 
Given    a $\{\circ, \bullet\}^{\Z}$-valued r.v. $Z_{0}$, such that $U\perp Z_{0}$, define the sequence of $\{\circ, \bullet\}^{\Z}$-valued r.v's $(Z_n)_{n\in \N}$ as follows:
\[Z_{n+1} = \cC(Z_{n}, U_{n}) \:.\]
Then $(Z_n)_{n\in \N}$ is a realization of model $C$. 

Again, the continuous time version of model $C$, with exponential holding times, is called a \textit{coalescing random walk} (cf. \cite[Ch. 2, Sec. 9]{grif79}).

\begin{figure}[Htb]
\begin{center}
\includegraphics[scale=0.65]{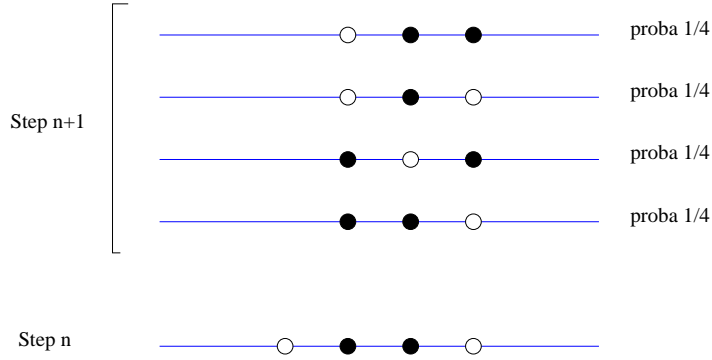}
\caption{The transition function of model $C$.}\label{fi-modelC}
\end{center}
\end{figure}

\subsection{Links between models $A$, $B$, and $C$}

One-step transition of the model $B$, resp. $C$, defines the mapping 
\begin{align*}
F_B: \quad \cM(\{\circ, \bullet\}^{\Z}) & \longrightarrow  \cM(\{\circ,
\bullet\}^{\Z}) \\
\mu & \longmapsto  \mu F_B \:,
\end{align*}
respectively, 
\begin{align*}
F_C: \quad \cM(\{\circ, \bullet\}^{\Z}) & \longrightarrow  \cM(\{\circ,
\bullet\}^{\Z}) \\
\mu & \longmapsto  \mu F_C \:.
\end{align*}
Define 
\begin{align*}
\varphi: \quad \{0,1\}^\Z & \longrightarrow  \{\circ, \bullet\}^\Z \\ 
 (x_i)_{i\in \Z} & \longmapsto  (y_i)_{i\in \Z} \:,
\end{align*}
with 
\[
y_i = \begin{cases} \bullet & \mbox{if } x_ix_{i+1} \in \{ 00,11\} \\
 \circ & \mbox{if } x_ix_{i+1} \in \{ 01,10 \} 
\end{cases}
\]
By extension, define  $\varphi: \ \cM(\{0,1\}^\Z) \rightarrow \cM(
\{\circ, \bullet\}^\Z)$. 
 
\begin{lemm}\label{le-commut}
The diagram below is commutative:
\begin{equation*}
\xymatrix{
\cM(\{0,1\}^\Z)  \ar[d]^{\varphi}  \ar[r]^-{F_A}  &
\cM(\{0,1\}^\Z)  \ar[d]^{\varphi} \\ 
\cM(\{\circ, \bullet\}^\Z)  \ar[r]^-{F_B}      & \cM(\{\circ, \bullet\}^\Z)
}
\end{equation*}
If $(X_n)_{n\in \N}$ is a realization of the Markov chain $A$, then
$(\varphi(X_n))_{n\in \N}$ is a realization of the Markov chain $B$. 
\end{lemm}

\begin{proof}
Recall that $\cA$ and $\cB$ are defined in (\ref{eq-A}) and (\ref{eq-B})
respectively. We are going to prove that:
\begin{equation}\label{eq-toprove}
\varphi \circ \cA = \cB \circ (\varphi, \mbox{Id}) \:. 
\end{equation}
The statement of the lemma follows. Set
\[
(x_i)_i,(u_i)_{i} \stackrel{\cA}{\longmapsto} (\tilde x_i)_{i}
\stackrel{\varphi}{\longmapsto} (\tilde y_i)_i, \qquad  (x_i)_i,(u_i)_{i}
\stackrel{\varphi,\mbox{Id}}{\longmapsto} (y_i)_{i}, (u_i)_i
\stackrel{\cB}{\longmapsto} (\hat{y}_i)_i\:.
\]
To obtain (\ref{eq-toprove}), it is enough to check that
$\tilde y_0=\hat{y}_0$. This is done by systematic 
inspection in Table \ref{tab:commute}. Each one of the 32 cases mimicks the
commutative diagram: in the first line, from left to right,
$(x_{-2},x_{-1},x_0)$, $(u_{-1},u_0)$, and $(\tilde x_{-1},\tilde x_0)$; in the second line,
from left to right, $(y_{-1},y_0)$, $(u_{-1},u_0)$, and
$\tilde y_0=\hat{y}_0$. 

\begin{table}[hbt]
\centering
\begin{tabular}{|c|c|c||c|c|c|}
\hline 
111 or 000&$\rightarrow\rightarrow$&11 or 00&101 or 010&$\rightarrow\rightarrow$&10 or 01
\\
\hline 
$\bullet\bullet$&$\rightarrow\rightarrow$&$\bullet$&$\circ\circ$&$\rightarrow\rightarrow$&$\circ$
\\
\hline 
\hline 
111 or 000&$\rightarrow\uparrow$&10 or 01&101 or 010&$\rightarrow\uparrow$&10 or 01
\\
\hline 
$\bullet\bullet$&$\rightarrow\uparrow$&$\circ$&$\circ\circ$&$\rightarrow\uparrow$&$\circ$
\\
\hline 
\hline 
111 or 000&$\uparrow\rightarrow$&01 or 10&101 or 010&$\uparrow\rightarrow$&10 or 01
\\
\hline 
$\bullet\bullet$&$\uparrow\rightarrow$&$\circ$&$\circ\circ$&$\uparrow\rightarrow$&$\circ$
\\
\hline 
\hline 
111 or 000&$\uparrow \uparrow $&00 or 11&101 or 010&$\uparrow \uparrow $&10 or 01
\\
\hline 
$\bullet\bullet$&$\uparrow \uparrow $&$\bullet $&$\circ\circ$&$\uparrow \uparrow$&$\circ$
\\
\hline 
\hline 
110 or 001&$\rightarrow\rightarrow$&11 or 00&100 or 011&$\rightarrow\rightarrow$&10 or 01
\\
\hline 
$\bullet\circ$&$\rightarrow\rightarrow$&$\bullet$&$\circ\bullet$&$\rightarrow\rightarrow$&$\circ$
\\
\hline 
\hline 
110 or 001&$\rightarrow\uparrow$&11 or 00&100 or 011&$\rightarrow\uparrow$&11 or 00
\\
\hline 
$\bullet\circ$&$\rightarrow\uparrow$&$\bullet$&$\circ\bullet$&$\rightarrow\uparrow$&$\bullet$
\\
\hline 
\hline 
110 or 001&$\uparrow\rightarrow$&01 or 10&100 or 011&$\uparrow\rightarrow$&10 or 01
\\
\hline 
$\bullet\circ$&$\uparrow\rightarrow$&$\circ$&$\circ\bullet$&$\uparrow\rightarrow$&$\circ$
\\
\hline 
\hline 
110 or 001&$\uparrow \uparrow $&01 or 10&100 or 011&$\uparrow \uparrow $&11 or 00
\\
\hline 
$\bullet\circ$&$\uparrow \uparrow $&$\circ $&$\circ\bullet$&$\uparrow \uparrow$&$\bullet$
\\
\hline 
\end{tabular}
\vspace{0,2cm}
\caption{The 32 possible cases.}\label{tab:commute}
\end{table}
If the process $X$ is defined by \eqref{rec-A}, relation \eqref{eq-toprove} entails that the process $Y$, defined by $Y_{n}=\varphi(X_{n})$, satisfies relation \eqref{rec-B}.
\end{proof}


\begin{lemm}\label{le-domin}
Model $B$ is dominated by model $C$: for $x,u
\in \{\circ, \bullet\}^{\Z} \times \cU^{\Z}$,
\[
 \cB(x,u) \leq \cC(x,u)\:,
\]
where $\leq$ is the coordinate-wise product ordering on $\{\circ,\bullet\}^{\Z}$,
with $\circ \leq \bullet$. 
\end{lemm}

\begin{proof}
This can be checked directly on the definitions of $\cB$ and
$\cC$. Intuitively, particles are merged in model $C$, and annihilate
in model $B$. 
\end{proof}

\begin{lemm}
The following implications hold:
\begin{align*}
& [C \mbox{ is ergodic with invariant measure } \delta_{\circ^\Z}] \\
& \qquad \implies [B \mbox{ is ergodic with invariant measure }
  \delta_{\circ^\Z}] \\
& \qquad \qquad \iff [A \mbox{ is non-ergodic with invariant measure }
  (\delta_{(01)^{\Z}} +
\delta_{(10)^{\Z}})/2] \:.
\end{align*}
\end{lemm}

\begin{proof}
This is a direct consequence of Lemmas \ref{le-commut} and
\ref{le-domin}. 
\end{proof}

Therefore, in order to prove Theorem \ref{th-main}, it is sufficient to prove that model
$C$ is ergodic with invariant measure $\delta_{\circ^\Z}$. This is the
purpose of next section. 

\section{\bf{Model $C$ is ergodic}} 
\label{monotone}


\begin{lemm}\label{le-monot}
Model $C$ is monotone, that is: for $z\in \{\circ, \bullet\}^{\Z}$,
$\tilde z\in \{\circ, \bullet\}^{\Z}$, $u\in \cU^{\Z}$, 
\[
z \leq \tilde z 
\ \implies \ \cC(z,u) \leq \cC(\tilde z,u)\:,
\]
where $\leq$ is the coordinate-wise product ordering. 
\end{lemm}

\begin{proof}
It can be checked directly on the definition of  $\cC$. 
\end{proof}

With this monotonicity, to get the ergodicity, it is enough to prove that
$\delta_{\bullet^\Z} F_C^n \rightarrow \delta_{\circ^\Z}$.  Indeed, consider two realizations of model $C$, one, say $Z=(Z_n)_n$, that starts with all sites occupied, the other, say $\tilde Z=(\tilde Z_n)_n$, that starts with an arbitrary initial condition, their evolution using the same update process $U$. According to Lemma \ref{le-monot}, at any time
$n\in\N$,  $\tilde Z_n\le Z_n$.  

\medskip

From now on, we focus on the process $Z$. Recall that for each $n$,
$Z_n = (Z_{k,n})_{k\in \Z}$ is the state of the system at time $n$. 
The process $Z_{n}$
is stationary, i.e. invariant by translation, since $Z_{0}$, $U$, and
$\cC$ are invariant too.  Define
\begin{equation}\label{eq-density}
d_n = \mathbb P \left(Z_{k,n}=\bullet\right) = \mathbb
P\left(Z_{0,n}=\bullet\right) \:.
\end{equation}
This is the density of particles at time $n$.  The density $d_n$ can also be viewed as an evaluation of the distance between $Z_{n}$ and $\delta_{\circ^\Z}$. Indeed, for any finite subset $E$
of $\Z$, consider the Hamming distance on $\{\circ,\bullet\}^E$, and
denote by $\mathcal W_H$ the corresponding Wasserstein distance on $\mathcal M (\{\circ,\bullet\}^E)$. Setting $Z_{E,n}
=(Z_{k,n})_{k\in E}$, we have:
$\mathcal W_H\left(Z_{E,n}, \delta_{\circ^E}\right)\ =\ |E|\,d_n$. 

\begin{theo}\label{th-density}
Let $T$ be the time that a simple symmetric random walk on $\Z$ needs
to reach 2, starting from 0. We have
\begin{align}
d_n & \ = \  \pr{ T > 2n } \label{eq-density1} \\
    & \ = \ 4^{-n}\,{2n+1\choose n}  \label{eq-density2} \:.
\end{align}
In particular, $d_n \sim 2/\sqrt{\pi\,n}$, hence converges to 0 as $n$ grows. 
\end{theo}
In continuous time, when the particles perform a simple symmetric
random walk, Bramson $\&$ Griffeath \cite{BrGr} obtain the same
asymptotic behavior for $d_n$, up to a scaling factor, as expected. 

\begin{coro}\label{co-C}
Model $C$ is ergodic with unique invariant measure
$\delta_{\circ^\Z}$.
\end{coro}

\begin{proof}
We first prove (\ref{eq-density2}), assuming
(\ref{eq-density1}). 
Let $S=(S_{k})_{k\in \N}$ be a realization of the simple symmetric
random walk on $\Z$, starting from 0. Define $M_{k}= \max \{ S_i ,
0\leq i \leq k \}$, the maximum of the random walk at time $k$. Recall
that $T=\inf \{i\geq 0 \mid S_i =2 \}$. 
We have
\begin{align*}
\pr{T > 2n} & \ = \  \pr{M_{2n}\le 1} \ = \ 1-\pr{M_{2n}\ge 2} \\
& \ = \  1-\sum_{\ell\in\Z}\pr{M_{2n}\ge 2, \ S_{2n}=\ell} \\
&\ = \ 1-\pr{S_{2n}\ge 2}-\sum_{\ell\le1}\pr{M_{2n}\ge 2, \ S_{2n}=\ell}
\end{align*}
According to the reflection principle, for $\ell\le1$,
$\pr{M_{2n}\ge 2, \ S_{2n}=\ell}= \pr{S_{2n}=4-\ell}$. Therefore, 
\begin{align*}
\pr{T > 2n} & \ = \  1-\pr{S_{2n}\ge 2}-\pr{S_{2n}\ge 3} \\
& \ = \  1-\pr{S_{2n}\le -2}-\pr{S_{2n}\ge 3} \\
&\ = \ \pr{S_{2n}\in\ac{0,2}} \\
&\ = \ 4^{-n}\,\pa{{2n\choose n}+{2n\choose n+1}} \ = \ 4^{-n} {2n+1\choose
  n} \:.
\end{align*}
Using Stirling's formula, we get 
\begin{equation*}
4^{-n} {2n+1\choose
  n} \sim \frac{2}{\sqrt{\pi
    n}} \:.
\end{equation*}

\medskip

Now let us prove (\ref{eq-density1}). 
Recall that $Z=(Z_n)_{n\in \N}$ is a realization of model $C$ with
$Z_0=\bullet^{\Z}$. 
One can extend the definition of $Z$ via coupling from the past. 
Consider the i.i.d. r.v.'s $(U_{k,n})_{(k,n)\in\Z\times\Z}$ with
$\pr{U_{i,j}=\uparrow}=\pr{U_{i,j}=\rightarrow}=1/2$. For each $s\in
  \Z$, define $Z^{(s)}=(Z^{(s)}_n)_{n\geq s}$ by 
 \begin{align*}
 Z^{(s)}_s = \bullet^{\Z}, \qquad \forall n \geq s, \ Z^{(s)}_{n+1} = \cC ( Z^{(s)}_n,
 (U_{k,n})_{k\in \Z}) \:.
 \end{align*}
The starting time of  $Z^{(s)}$ is $s$, but, besides that, the
dynamic is the same as that of process $Z$. Observe that
$Z^{(0)}=Z$. More generally, $Z^{(s)}$ has the same
distribution as $\left(Z_{-s+n}\right)_{n\ge s}$. Thus, we have
\[
d_n = \pr{Z^{(-n)}_{0,0} = \bullet} \:.
\]
In Figure \ref{ancestry}, we have represented space-time diagrams for
the model. The point of coordinate $(k,n)$ corresponds to site $k$ at
step $n$. We have also represented the updating variables
with the following convention: at the point $(k,n)$, there is an
arrow pointing north if $ U_{k,n}=\uparrow$ and an arrow pointing
north-east if $ U_{k,n}=\rightarrow$. This allows to visualize the
evolution of particles in the processes $Z^{(s)}$. In Figures
\ref{ancestry}.b and
\ref{ancestry}.c, the processes $Z^{(-3)}$ and $Z^{(-5)}$ are
represented; the grey nodes are the ones whose color depend on
updating variables outside of the represented window. 
In Figure \ref{ancestry}.d, the particles
painted in orange (gray) are those that merged into the
particle present at time 0 and site 0. 

\medskip

Let $I_n$, $n\geq 0$, be the set of indices of particles present at time $-n$ in $Z^{(-n)}$
that merge into particle 0 at time 0. Either
$I_{n}=[\![a_{n},b_{n}]\!]$, $a_n \leq b_n$, in which case
$Z^{(-n)}_{0,0}=\bullet$, or $I_{n}=\varnothing$, in which case
$Z^{(-n)}_{0,0}=\circ$.
We focus on $|I_n|$. 
For instance, in Figure \ref{ancestry}.d, we have
$(|I_n|)_{n\in[\![0,6]\!]}=(1,2,1,2,3,4,3)$.
Observe that 
\begin{equation}\label{eq-two}
d_n = \pr{ Z^{(-n)}_{0,0} = \bullet} = \pr{|I_n| \geq
  1} \:.
\end{equation}

\begin{figure}[Htb]
\begin{center}
\includegraphics[height=4.5cm]{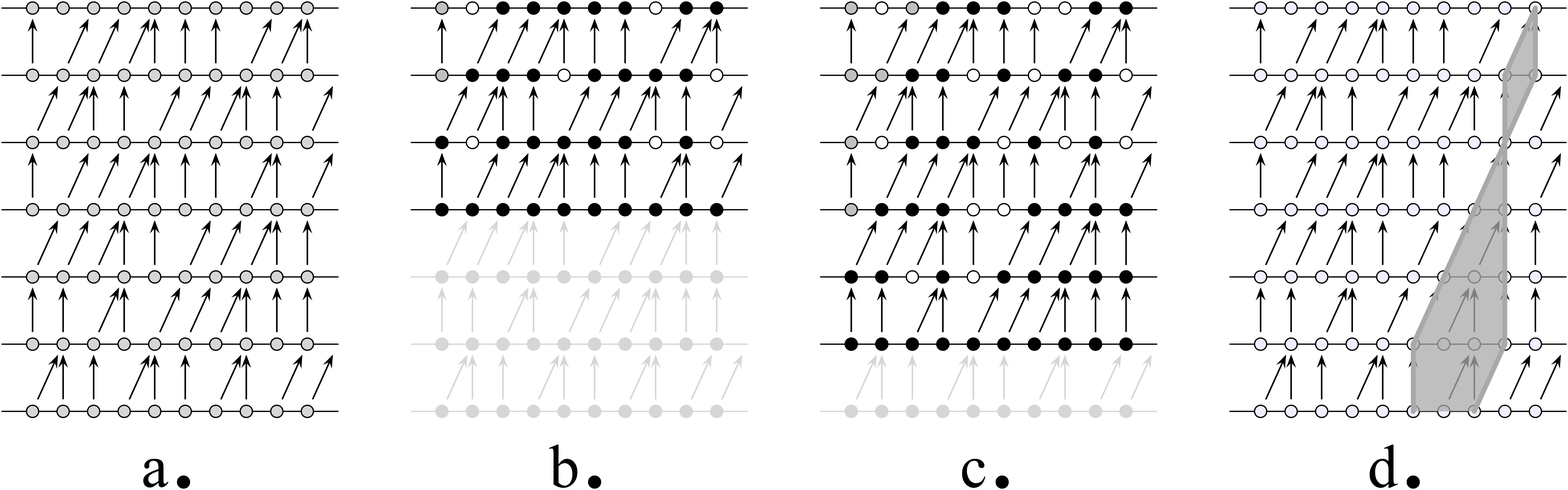}
\caption{\textbf{a.} The updating
  $(U_{k,n})_{(k,n)\in[\![-9,0]\!]\times[\![-6,0]\!]}$. \textbf{b.}
  The process $Z^{(-3)}$ during the span $[\![-3,0]\!]$. \textbf{c.}
  The process $Z^{(-5)}$ during the span $[\![-5,0]\!]$. \textbf{d.}
  The merging of particles.}\label{ancestry} 
\end{center}
\end{figure}

We have $a_{0}=b_{0}=0$, and $|I_0|=1$. Define 
\[\rho=\frac{1}{4}\delta_{-1}+\frac{1}{2}\delta_{0}+\frac{1}{4}\delta_{1}
\:.\]
We check the following:
\begin{itemize}
  \item Assume that $a_{n}<b_{n}$. Then, 
\[
a_{n+1} = \begin{cases} a_{n}-1 & \mbox{if } U_{a_n-1,-n-1}=
                       \rightarrow \\
a_n & \mbox{if } U_{a_n-1,-n-1}= \uparrow 
\end{cases}, \qquad b_{n+1} = \begin{cases} b_{n}-1 & \mbox{if }
  U_{b_n,-n-1}= \rightarrow \\
b_n & \mbox{if }
  U_{b_n,-n-1}= \uparrow 
\end{cases}\:.
\]
Thus, $|I_{n+1}| -|I_n| \in\ac{0,\pm
  1}$, and the conditional law of
$|I_{n+1}| -|I_n|$ is $\rho$. 
  \item Assume that  $a_{n}=b_{n}$. 
Then
\[
I_{n+1} = \begin{cases} \varnothing  & \mbox{if } (U_{a_n-1,-n-1},
  U_{a_n,-n-1}) = (\uparrow,\rightarrow) \\
[\![a_{n}-1,a_{n}-1]\!] & \mbox{if } (U_{a_n-1,-n-1},
  U_{a_n,-n-1}) =  (\rightarrow, \rightarrow) \\
[\![a_{n}-1,a_{n}]\!] & \mbox{if } (U_{a_n-1,-n-1},
  U_{a_n,-n-1}) = (\rightarrow, \uparrow) \\
[\![a_{n},a_{n}]\!] & \mbox{if } (U_{a_n-1,-n-1},
  U_{a_n,-n-1}) =  (\uparrow, \uparrow) \:.
\end{cases}
\]
For instance, the third case appears between lines $-3$ and $-2$ in
Figure \ref{ancestry}.d. 
Here again, the conditional distribution of
$|I_{n+1}| -|I_n|$ is $\rho$. 
  \item If $I_{n}=\varnothing$, then $I_{n+1}=\varnothing$.
\end{itemize}

Consequently $(|I_n|)_{n\in
  \N}$ is a random walk with step $\rho$, starting from 1, and killed
when it reaches 0. Using (\ref{eq-two}), we obtain (\ref{eq-density1}).
\end{proof}

\section{\bf{Speed of convergence for models $A$ and $B$}}\label{se-speedAB}

Let $(A_n)_{n\in \N}$ be a realization of model $A$, with $A_0 \sim \mu$, $\mu
\in \cM(\{0, 1\}^{\Z})$. The possible limits for weakly-converging
subsequences of  $(A_n)_{n\in \N}$ are of the form $p\delta_{(01)^\Z}
+ (1-p)\delta_{(10)^\Z}$ for $p\in [0,1]$.   An evaluation of the distance to the limits is given
by
\[
\pr{ A_{0,n}A_{1,n} \in \{00,11\} } \:.
\]
Since model $A$ is not monotone,
we do not know for which
initial measure $\mu$ this distance will be maximized. Hence we
evaluate the ``speed of convergence'' for model $A$ by the quantity:
\[
d_n^{A} = \max_{\mu\in \cM(\{0, 1\}^{\Z})} \pr{ A_{0,n}A_{1,n} \in
 \{00,11\} } \:.
\]
 The quantity $d_n^{A}$ is also the speed of convergence to
$\delta_{\circ^{\Z}}$ for model $B$. 
Indeed we have $\pr{ A_{0,n}A_{1,n} \in \{00,11\} } = \pr{ \varphi(A_{n})_0 =
\bullet}$, which implies that
\[
d_n^{A} = \max_{\nu\in \cM(\{\circ, \bullet\}^{\Z})} \pr{B_{0,n} =
\bullet} \:,
\]
where $(B_n)_n$ denotes a realization of model $B$ and $\nu$ denotes
its initial distribution ($B_0\sim \nu$). 

\medskip

Recall that $d_n=4^{-n} {2n+1 \choose n}$ is the speed of
convergence for model $C$, see (\ref{eq-density}) and Theorem \ref{th-density}.

\begin{prop}\label{pr-coupling}
We have
\[
\frac{1}{2} d_{n-1} \leq d_n^{A} \leq d_n \:.
\]
\end{prop}

Let $(A_n)_n$ be a realization of model $A$, with $A_0 \sim \mu$.
If $\mu$ is the uniform distribution on $\{0,1\}^{\Z}$,
i.e. the r.v.'s $A_{i,0}$ are i.i.d. with
$\pr{A_{0,0}=0}=\pr{A_{0,0}=1}=1/2$, then we shall see that $\pr{ A_{0,n}A_{1,n} \in
 \{00,11\} } = d_n/2$. If $\mu=\delta_{1^\Z}$, then we have
$\pr{ A_{0,n}A_{1,n} \in
 \{00,11\} } = d_{n-1}/2$, which is larger than $d_n/2$. The results can be translated to model $B$: the density of particles
at step $n$ is $d_n/2$ if the initial distribution is uniform, and it
is $d_{n-1}/2$ if the initial distribution is
$\delta_{\bullet^{\Z}}$.

\medskip

The end of the section is devoted to the proof of
Prop. \ref{pr-coupling}, through the study of model $A$ with the two initial distribution mentioned previously. 

\medskip

We define a new PCA, called model $D$, which is a coupling of models $B$ and $C$.
The alphabet is $\{\circ, b, g\}$ and the set of sites is $\Z$. Given
a configuration $u\in \{\circ, b, g\}^{\Z}$, the interpretation is as
follows: if $u_i=\circ$ then site $i$ is empty; if $u_i = b$ then site
$i$ contains a {\em blue} particle; if $u_i = g$ then site
$i$ contains a {\em green} particle. Particles move as in models $B$
and $C$. When two particles collide,
they get merged into one particle as in model $C$. In absence of
collision, particles keep their color. In case of a collision, 
the merged particle is colored according to the rules:
\begin{equation}\label{eq-merge}
b+b \rightarrow g, \quad g+g \rightarrow g, \quad b+g \rightarrow b,
\quad g+b \rightarrow b \:.
\end{equation}
We have represented a realization of model $D$ on Figure
\ref{colors}. The ``question mark'' nodes are the ones whose color depend on 
updating variables outside of the represented window. 
Define
\begin{align*}
\pi_B:&  \{\circ, b,g\} \longrightarrow \{\circ,\bullet\}, & 
 \pi_B(\circ)=\circ, \ \pi_B(b)=\bullet, \pi_B(g)=\circ \:, \\
\pi_C:&  \{\circ, b,g\} \longrightarrow \{\circ,\bullet\}, & 
 \pi_C(\circ)=\circ, \ \pi_C(b)=\bullet, \pi_C(g)=\bullet \:.
\end{align*}
We keep the same notations for the product applications: $\pi_B:
\{\circ, b,g\}^{\Z} \rightarrow \{\circ,\bullet\}^{\Z}, (u_i)_i
 \mapsto (\pi_B(u_i))_i$, and  $\pi_C:
\{\circ, b,g\}^{\Z} \rightarrow \{\circ,\bullet\}^{\Z}, (u_i)_i
 \mapsto (\pi_C(u_i))_i$.

\begin{lemm}\label{le-coupling}
If $(D_n)_n$ is a realization of model $D$, then $(\pi_B(D_n))_n$ is a
realization of model $B$, and $(\pi_C(D_n))_n$ is a
realization of model $C$. As consequences, model $D$ is ergodic with
unique invariant measure $\delta_{\circ^{\Z}}$, and $d_n^{A} \leq d_n$.
\end{lemm}

\begin{figure}[Htb]
\begin{center}
\includegraphics[width=7cm]{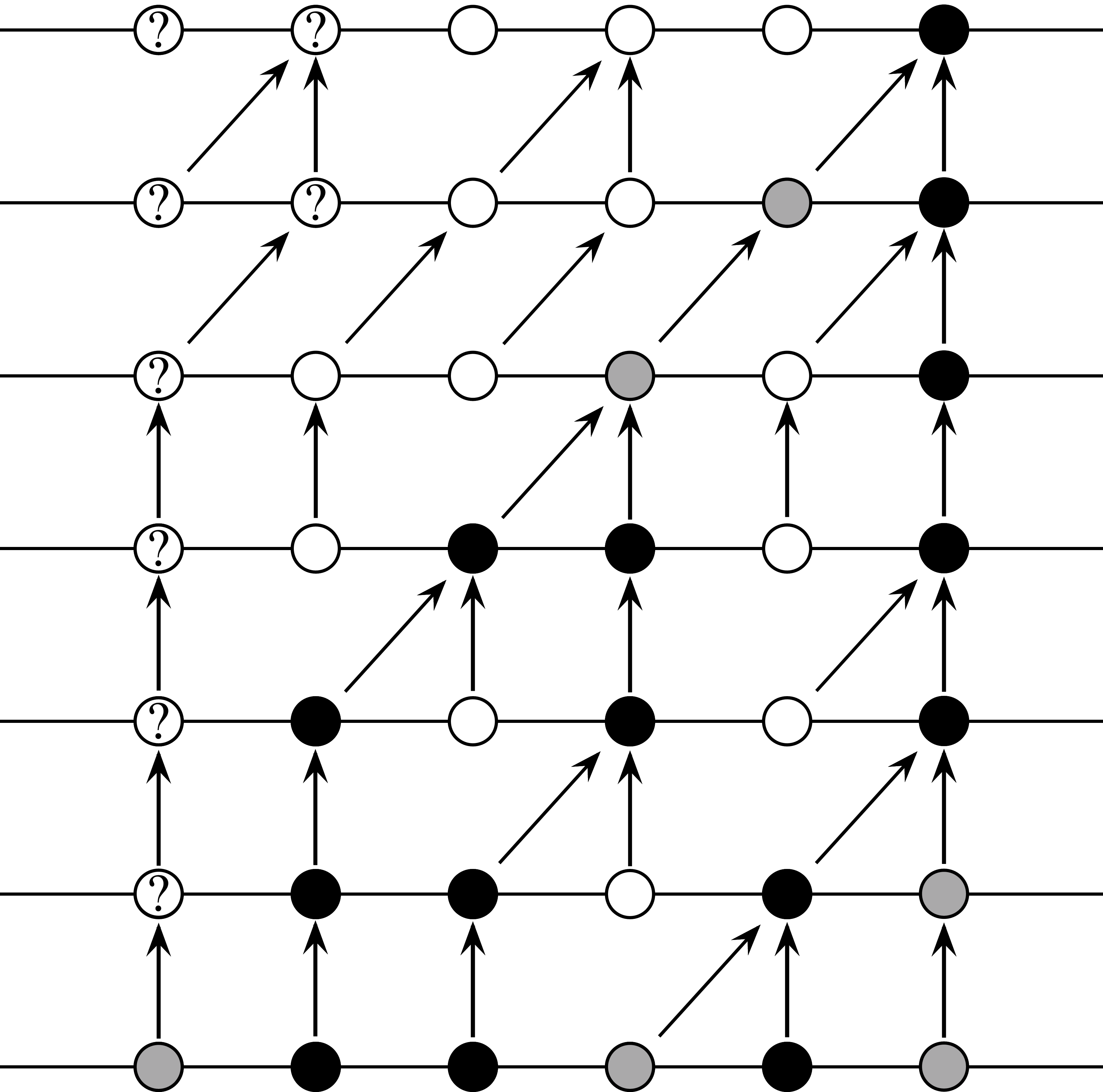}
\caption{A realization of process $D$ (gray for g(reen),
 black for  b(lue)).}\label{colors}
\end{center}
\end{figure}

Let $(D_n)_n$ be a realization of model $D$ with $D_0$ being defined as follows: the r.v.'s $D_{i,0}$ are i.i.d. with $\pr{D_{0,0}=b}=\pr{D_{0,0}=g}=1/2$. At step $n$, the colors of the remaining particles are still i.i.d. and uniform : whatever the shape of the binary  tree  of coalescences leading to the presence of a particle at a given position at time $n$ (see Figure \ref{colors} for an example), if the colors of the initial particles are independent and if one of these particle's color is uniform, then, due to \eqref{eq-merge}, the resulting color will still be uniformly distributed.  
Therefore we have 
\begin{equation}\label{eq-iid}
\pr{\pi_B(D_{0,n})=\bullet} =\pr{D_{0,n} = b}= \frac{1}{2}
\pr{D_{0,n} \in \{b,g\}} = \frac{1}{2} \pr{\pi_C(D_{0,n})=\bullet} =
\frac{1}{2} d_n\:,
\end{equation}
where the last equality follows from Theorem \ref{th-density}.

\medskip

Now let $(\widetilde{D}_n)_n$ be a realization of model $D$ with
$\widetilde{D}_0=b^{\Z}$. 
Define $E=(E_{i})_i$ by
\[
E_{i} = \begin{cases} b & \mbox{if } \widetilde{D}_{i,1}=b \\
g & \mbox{if } \widetilde{D}_{i,1}= \circ \mbox{ or } g
\end{cases}\:.
\]
The r.v.'s $(E_{i})_i$ are i.i.d. with $\pr{E_{0}=b}=
\pr{E_{0}=g}=1/2$. Let us justify this point. The state at time 1 of a realization of model $A$ that
starts from $0^{\Z}$ is uniformly distributed by definition. Hence,
the state at time 1 of a realization of model $B$ that
starts from $\bullet^{\Z}$ is uniformly distributed. And $E$ has the
same law as the latter up to the transformation $b\leftrightarrow \bullet,
\ g \leftrightarrow \circ$. 

So we have $E\sim D_0$.  
Observe also that
$\pi_B(\widetilde{D}_1)= \pi_B(E)$. 
We deduce that, for all $n\geq 1$, we have 
$\pi_B(\widetilde{D}_n) \sim \pi_B(D_{n-1})$. In particular, using (\ref{eq-iid}),
\[
\pr{\pi_B(\widetilde{D}_{0,n})=\bullet} = \pr{\widetilde{D}_{0,n}=b}=
\pr{D_{0,n-1}=b} = d_{n-1}/2\:. 
\]
This completes the proof of Prop. \ref{pr-coupling}.

\subsection*{Conclusion}

The following question remains: does there exist a {\em
  positive-rates} PCA which is non-ergodic with a unique invariant
measure? 

Let us provide some context. By definition, a PCA  has {\em
  positive-rates} if all its probability transitions are different
from 0 and 1 (more formally, if $f: \Sigma^V \rightarrow \cM(\Sigma)$
is the transition function, then $\forall u \in \Sigma^V,\forall v\in
\Sigma, \ f(u)(v) \in (0,1)$). 
It had been a long standing
conjecture that all 1-dimensional positive-rates PCA are
ergodic. In \cite{gacs}, G\'{a}cs disproved the conjecture by exhibiting a
complex counter-example with several invariant measures. The existence
of the intermediate case (unique but non-attractive invariant measure)
remains open. A priori, it is not possible to perturbate model $A$
to get a positive rates example.




\end{document}